\documentclass[pdflatex,sn-nature]{sn-jnl}

\usepackage{array}
\usepackage{graphicx}%
\usepackage{multirow}%
\usepackage{amsmath,amssymb,amsfonts}%
\usepackage{amsthm}%
\usepackage{mathrsfs}%
\usepackage[title]{appendix}%
\usepackage{xcolor}%
\usepackage{textcomp}%
\usepackage{manyfoot}%
\usepackage{booktabs}%
\usepackage{algorithm}%
\usepackage{algorithmicx}%
\usepackage{algpseudocode}%
\usepackage{listings}%
\usepackage{caption}

\theoremstyle{thmstyleone}%

\newtheorem{corollary}{Corollary}
\newtheorem{lemma}{Lemma}
\newtheorem{theorem}{Theorem}
\theoremstyle{thmstyletwo}%

\theoremstyle{thmstylethree}%
\newtheorem{definition}{Definition}%

\begin{document}

\title[Article Title]{Biochemical Computing Mode for Sequential Logic}


\author[1,2]{\fnm{Han} \sur{Huang}}\email{hanhuang@foxmail.com}
\equalcont{These authors contributed equally to this work.}

\author[3]{\fnm{Chengzhi} \sur{Ma}}\email{se21mcz@mail.scut.edu.cn}
\equalcont{These authors contributed equally to this work.}

\author*[3]{\fnm{Yuxin} \sur{Zhao}}\email{yuxinzhaozyx@163.com}

\author*[3]{\fnm{Qingyao} \sur{Wang}}\email{Qingyao0029@163.com}
 
\author[4]{\fnm{Xinglong} \sur{Xiao}}\email{fexxl@scut.edu.cn}

\author[5,6]{\fnm{Xiulin} \sur{Shu}}\email{shuxl@gdim.com}

\author[7]{\fnm{Zhifeng} \sur{Hao}}\email{haozhifeng3377@163.com}

\affil[1]{\orgdiv{Key Laboratory of Big Data and Intelligent Robot}, \orgname{MOE of China},  \orgaddress{\city{Guangzhou}, \postcode{510006}, \country{China}}}

\affil[2]{\orgdiv{Key Laboratory of Symbolic Computation and Knowledge Engineering of Ministry of Education}, \orgname{Jilin University}, \orgaddress{ \city{Changchun}, \postcode{130012},  \country{China}}}

\affil*[3]{\orgdiv{School of Software Engineering}, \orgname{South China University of Technology}, \orgaddress{\city{Guangzhou}, \postcode{510006},  \country{China}}}

\affil[4]{\orgdiv{School of Food Science and Engineering}, \orgname{South China University of Technology},\orgaddress{\city{Guangzhou}, \postcode{510006},  \country{China}}}

\affil[5]{\orgdiv{Institute of Microbiology}, \orgname{Guangdong Academy of Sciences},\orgaddress{\city{Guangzhou}, \postcode{510070},  \country{China}}}

\affil[6]{\orgdiv{State Key Laboratory of Applied Microbiology Southern China}, \orgaddress{\city{Guangzhou}, \postcode{510070},  \country{China}}}

\affil[7]{\orgdiv{Department of Mathematics}, \orgname{Shantou University},\orgaddress{\city{Shantou}, \postcode{515063},  \country{China}}}


\abstract{Recent years have witnessed the growing scholarly interest in the next-generation general-purpose computers. Various innovative computing modes have been proposed, such as optical, quantum phenomena, and DNA-based modes. Sequential logic circuits are a critical factor that enables these modes to function as general-purpose computers, given their essential role in facilitating continuous computation and memory storage through their ability to store states. However, compared to computability, it is often overlooked due to the difficulty of its implementation. In this paper, we first demonstrate sequential mapping, a crucial necessary condition for electronic computers to realize sequential logic circuits, and highlight this distinctive property of general-purpose computers in the context of logic gate circuits. To achieve computational functionalities comparable to those of electronic computers, we utilize the control effect of enzymes on enzymatic reactions to design a logic gate model that is composed of small molecules and driven by enzymes, subsequently propose a biochemical computing mode. Furthermore, we mathematically analyze the static and dynamic input-output properties of biochemical logic gate components and prove that the biochemical computing mode satisfies sequential mapping similar to electronic computers. When combined with the storage characteristics of NOT-AND gates, it can realize sequential logic circuits. The findings can serve as a theoretical foundation for developing general-purpose biochemical computers.}




\maketitle

\section{Introduction}\label{sec1}
The advent of electronic computers has exerted a profound impact on human life, accelerating societal and scientific development. In 1946, ENIAC, the first electronic computer, was introduced. However, it relied on temporary storage, requiring all switches to be reset for each computation. In the same year, Neumann proposed the “stored-program computer” \cite{von1993first}, which significantly improved computational efficiency and endowed electronic computers with sequential mapping. The von Neumann architecture laid the groundwork for modern computer architecture, along with hardware and software ecosystems. As physical hardware and computer architectures have evolved, electronic computers have undergone four critical phases: vacuum tubes, transistors, large-scale integration, and ultra large-scale integration integration. At present, due to limitations in manufacturing processes and power consumption, the performance growth of electronic computers is nearing a bottleneck, and Moore’s Law \cite{markov2014limits} has begun to falter. Consequently, next-generation computing has become a hot topic, resulting in novel computing modes such as optical, biological, and quantum computing. Studies have begun to investigate their potential in realizing the next-generation computers.

At the dawn of electronic computing, computers were merely regarded as enormous calculators devoid of the concept of time. However, as computers evolved and task complexity increased, scientists had to consider the behavior of systems over time, leading to the introduction of sequential logic circuits. According to digital circuit theory, a sequential logic circuit refers to a circuit whose steady-state output at any moment depends not only on the current input but also on the state formed by the previous input. In contrast, the output of combinational logic is solely a functional relationship with the current input \cite{katz2005contemporary}. In other words, sequential logic circuits incorporate the concept of time through the use of memory, significantly enhancing the versatility of computers. Scientists further designed the logical structure of computers to realize general-purpose computers. The logical structure proposed by Neumann \cite{burks1946preliminary} utilized the reusability of electronic logic gates and employed the concept of program storage to separate hardware from software. This separation allows programmers to continuously input instructions, enabling the computer to operate through logic circuits without the need to rearrange circuit boards or reset switches. Until now, the principle of stored programs remains a fundamental aspect of computer operation, and many new computing modes continue to base their modular designs on this principle.  As shown above, the application of sequential logic circuit has significantly advanced the development of general-purpose computers.

There are numerous hypotheses regarding the computing modes for next-generation general-purpose computers. Among them, the most developed and mainstream modes include optical computing \cite{mcmahon2023physics}, DNA computing based on biochemical reactions \cite{fan2020propelling}, and quantum computing grounded in quantum mechanics \cite{zhu2024quantum}. We will introduce each of these modes individually in the following.

Optical computing utilizes photons as carriers for information transfer, which replaces electrons with photons and thus enables optical computation instead of electronic computation. Studies on optical computing began as early as the 1950s \cite{ambs2010optical}, when researchers explored the use of optical devices for information processing and computation. Later in 1963, advancements in technologies such as lasers and optical fibers led researchers to focus on constructing high-performance general-purpose optical computers. Due to the faster speed of photon compared to electrons, optical computers can outperform electronic computers in terms of speed, while also offering advantages such as enhanced fault tolerance, higher processing precision, and lower energy consumption. Currently, researchers primarily leverage the ample dissipation and parallelism of photons to address NP-complete problems \cite{xu2020scalable,vazquez2018optical}. Compared to electronic computers, photon-based computers significantly reduce the time required to solve such problems. Furthermore, researchers have applied optical computation to the training of neural networks and decreased energy consumption \cite{wright2022deep}. In addition, optical storage materials \cite{burks1946preliminary} have also become a research focus. The current realization of optical computation relies on the high-speed and high-precision conversion between optical and electrical signals. If the storage of photon information can be achieved, it would enable the independent execution of computation and storage processes by photons, eliminating the need for intermediate electrical signal conversion and facilitating efficient pure optical signal computation. However, current studies on optical computing focus more on specialized analog signal computing rather than general-purpose Boolean logic computing  \cite{mcmahon2023physics}, since the performance of optical computing elements often lags behind that of electronic computing elements.

DNA computing is a mode that utilizes DNA molecules and associated chemical substances as fundamental materials, relying on biochemical reactions for its operations. In 1994, Adleman first proposed the DNA computing model 
 \cite{adleman1994molecular}, which enabled access to this innovative field. In 2004, Okamoto et al. \cite{okamoto2004dna} first integrated digital circuits with DNA computing to construct DNA circuits, forming DNA logic gates that combine different levels of logic gates in a cascading manner to create complex circuits. In 2009, Li et al. \cite{Li2001dna} explored the fundamental principles of DNA computing. Current improvements in DNA logic gate models focus primarily on reaction speed and reusability. For instance, Song et al. \cite{song2019fast} introduced a DNA logic gate model with faster reaction rates and fewer DNA strands. Chatterjee et al. \cite{chatterjee2017spatially} enhanced the likelihood of strand collisions by constraining the reaction space of DNA. Song et al. \cite{song2017renewable} proposed a reusable DNA toggle logic circuit where the toggle reaction is reversible. However, due to the irreversibility of circuit components, the toggle gates remain non-reusable. Additionally, Lv et al. \cite{lv2023dna} designed a programmable DNA integrated circuit that can realize over ten billion different circuits through the programming of 24 addressable DNA logic gates. Compared to electronic computing, DNA computing has inherent advantages such as its parallel processing capability, along with high energy efficiency and large storage capacity \cite{Shan2021dna}. Current research on DNA computing focuses on the derivation of biochemical modes \cite{vishweshwaraiah2021two}, applications of DNA computing \cite{ravichandran2021efficient}, and the realization of DNA computers. In addition, only a few studies have investigated its sequential logic, such as constructing pseudo-bistable structures using multiple strand displacement reactions \cite{burks1946preliminary} to satisfy sequential logic.

 Quantum computing is a novel computing mode based on quantum mechanics theory. In 1982, American physicist Feynman first introduced the concept of “quantum simulation” \cite{feynman2018simulating}, which involves using devices that operate according to quantum mechanics to simulate the evolution of quantum systems. This was followed by British physicist Deutsch’s proposal of the “quantum Turing machine” \cite{deutsch1985quantum}, which can be equivalently represented as a quantum circuit model. Since then, research on quantum computers has aroused scholarly interest, which leads to the emergence of various models of quantum computing \cite{chamberland2022building,alexeev2021quantum}. Compared with electronic computing, quantum computing has two main advantages: 1) quantum parallelism and 2) quantum entanglement. These two characteristics offer exponential speed advantages to quantum algorithms compared to classical algorithms. Currently, relevant research focuses on quantum algorithms, applications, and physical realization of quantum computing. For example, in terms of physical realization, Chamberland et al. proposed the use of cascaded cat codes to construct fault-tolerant quantum computers, which achieved high fidelity for Toffoli gates with lower qubit costs \cite{chamberland2022building}. Youngseok et al. \cite{kim2023evidence} conducted experiments on a quantum processor with 127 qubits to investigate the noise resistance and fault tolerance of quantum computers. In recent years, Pezzagna et al. suggested using nitrogen-vacancy centers in diamonds as effective room-temperature qubits for diamond-based quantum computers  \cite{pezzagna2021quantum}. At present, quantum computers outperform traditional electronic computers in solving quantum simulation problems. For example, Gabriel et al. used fermionic quantum turbulence to break through the limits of high-performance computing, enabling desktop quantum experiments to simulate complex dynamics  \cite{wlazlowski2024fermionic}. They also have practical applications in materials development for batteries and industrial catalysts \cite{daley2022practical}. Overall, studies on quantum computers tend to focus on material selection, improvements in quantum algorithms, as well as noise and fault tolerance. However, few studies have mentioned sequential logic, which is one of the key elements in constructing a general-purpose quantum computer.

 Compared to existing electronic computing modes, these novel modes offer significant advantages in parallelism, which in most cases enables them to achieve higher efficiency when tackling problems with super-polynomial time complexity. Studies in various fields have proposed numerous hypotheses and models to achieve next-generation computers. However, limited research has focused on the fulfillment of sequential logic, suggesting the need for further investigation. For instance, optical memory devices are among the difficult optical components to realize. Recent studies have combined silicon photonics with phase-change materials to develop non-volatile integrated optical memories, though issues related to fabrication difficulties and stability persist \cite{chen2023neuromorphic}. In DNA computing, The inability of reactants or circuit components to change states multiple times hinders operations similar to electronic logic gates, posing a challenge to achieving temporal consistency \cite{song2019fast}. Although quantum computing has considered the design and implementation of quantum registers, it still struggles with maintaining long-term and accurate storage similar to electronic registers \cite{bradley2019ten}. Sequential logic is one of the essential conditions for constructing classical computers, and it enables classical computers to perform continuous computations and achieve storage without the need to reset all switches before each computation, as was required in the first electronic computers \cite{burks1981first}. Therefore, to realize a sustainable and user-friendly general-purpose computer, the prerequisite is to equip it with sequential logic.

 In this paper, inspired by Adleman’s insights into the potential of molecular scientific computing  \cite{adleman1998computing}, and in light of the non-consuming, specific, and highly efficient catalytic properties demonstrated by enzymes in molecular computing \cite{chaplin1990enzyme}, in order to achieve sequential logic, we propose the property of sequential mapping, a critical requirement to achieve this logic. Based on these insights, we explore the potential of biochemical computing modes, leveraging the unique catalytic properties of enzymes to design a novel approach that achieves sequential logic and meets the requirements of a general-purpose computer. Furthermore, we demonstrate the feasibility of biochemical computing modes in realizing general-purpose computers by fulfilling the property of sequential mapping.

The main contributions of the paper are:
\begin{itemize}
    \item{\textbf{Sequential mapping.} We propose the concept of sequential mapping and prove that electronic computers satisfy sequential mapping. Furthermore, we reveal that sequential mapping is a necessity for sequential logic, which is crucial for constructing general-purpose computers.}
    \item{\textbf{Biochemical computing mode.} We propose a new computing mode, namely biochemical computing mode, which utilizes the control capabilities of enzymes over biochemical reactions to achieve sustainable reuse of logic gates. Additionally, we construct reusable biochemical logic gate circuits and realize biochemical combinational logic, which exhibits potential advantages such as non-consumption, specificity, and high-efficiency catalysis.}
    \item{\textbf{Sequential logic based on biochemical computing mode.} We demonstrate that the biochemical logic gate circuits based on biochemical computing mode satisfy the sequential mapping property, thereby realizing sequential logic. This proves the feasibility of biochemical computing mode to become a general-purpose computer, similar to electronic computers.}
\end{itemize}

The remainder of this paper is organized as follows. In Section 2, we introduce the definition of sequential mapping and prove that electronic computers possess sequential mapping that meets this definition. Section 3 presents biochemical computing mode, designs and constructs basic logic gate models based on enzymatic reactions and biochemical computation, including the construction and implementation of biochemical NOT, AND, and OR gate models. Section 4 analyzes the dynamic input and output characteristics of biochemical logic gates based on biochemical computing mode, proving that biochemical computing mode satisfies sequential mapping and realizes sequential logic. The conclusion and future prospects of biochemical computing mode are given in Section 5.

\section{Sequential Mapping of Electronic Computers}\label{sec2}
Sequential logic is a crucial prerequisite for achieving the universality of computers. Constructing a sequential logic circuit \cite{katz2005contemporary} necessitates two major conditions, one of which is the requirement for steady-state output, i.e., the ability to stably output a signal within a given period of time. To investigate this necessary condition mentioned previously, we propose a formal definition of sequential mapping and explicate its intrinsic relationship with sequential logic. Subsequently, the sequential mapping of electronic logic circuits is systematically examined.

\subsection{Sequential Mapping and Logic Circuits}
Due to the different characteristics of various computational carriers, the computational functionality implemented by their corresponding circuits cannot be guaranteed to be perfectly equivalent to the ideal computing results. For example, due to its physical properties, the electronic circuits of electronic computers often have certain errors in computing results compared with the intended computational functionality. We believe that the nature of such errors directly determines the universal computing capability of the computing carrier. To study this nature, we propose the characteristic of sequential mapping, which is defined as follows:

\begin{definition}[Sequential mapping]
Given a specified positive real number $\kappa$, a circuit $S$ and a sequential mapping function $f_S$, the circuit $S$ is said to satisfy sequential mapping if there exists a positive real number $\tau_S$ such that for any input signals $X_1,\dots,X_n$ and their sequential mapping functions $f_{X_1},\dots,f_{X_n}$, the output of the circuit $S(X_1,\dots,X_n,t)$ satisfies: $\forall t  (  (|S(X_1,\dots,X_n,t)-f_S(f_{X_1},\dots,f_{X_n},t)|>\kappa) \rightarrow (|S(X_1,\dots,X_n,t+\tau_S)-f_S(f_{X_1},\dots,f_{X_n},t+\tau_S)|<\kappa)  ) $.
\end{definition}
Sequential mapping function $f_S$ represents the intended
computational functionality of the circuit $S$ and must be sustained over an extended period according to the definition of sequential mapping.  $\tau_S$ indicates the total computational delay of the circuit $S$ and its input signal, $f_S(f_{X_1},\dots,f_{X_n},t)$ denotes the expected output of the circuit $S$ at time $t$, and $S(X_1,\dots,X_n,t)$ represents the actual output of the circuit at that time. This definition implies that, under the constraints of a given error $\kappa$ and delay $\tau_S$, if a computational circuit $S$ can realize the function $f_S$ and meet the sequential mapping, then the output of $S$ will maintain its deviation from the sequential mapping function $f_S$ within the error margin $\kappa$ over a long period. Moreover, the definition allows for short-term deviations beyond $\kappa$, provided that the output of the circuit returns to within the acceptable range of $\kappa$ in at most $\tau_S$ time units.

Sequential mapping necessitates that circuits preserve correct computational functionality at any time $t$. A sequential logic circuit is defined as a system whose steady-state output at any moment depends not only on the current inputs but also on the previous input states \cite{sep-turing-machine}. While sequential logic circuits enable the implementation of complex digital systems—such as state machines, memory units, and counters—their design cannot rely exclusively on sequential mapping. Instead, they require cascaded logic gates to achieve full functionality. Since OR and NOT gates are functionally complete (i.e., capable of constructing any combinational logic circuit), they constitute fundamental requirements for the operational validity of sequential logic systems. To formalize this principle, we propose a general theorem.

 \begin{theorem}
Circuits that satisfy sequential mapping while exhibiting the computing capability of NOT-AND gate effectively realizes sequential logic circuits.
\label{theorem-1}
\end{theorem}
\begin{proof}
    Since a circuit with NOT-AND gate computing capabilities can construct any Boolean logic circuit through cascading, we can build an RS flip-flop\cite{katz2005contemporary}, which preserves $t$-th Boolean signals from $t$-th inputs and $(t-1)$-th Boolean signals, using cross-coupled NOT-AND gates. Let circuit $S$ be composed of two NOT-AND gates with their inputs and outputs cross-coupled, where the inputs are $X_1$ and $X_2$. Therefore, the intended computational functionality of the circuit $S$ is denoted as $f_S(f_{X_1},f_{X_2},t)=X_1(t)+X_2(t)f_S(f_{X_1},f_{X_2},t-1)$ for the sake of generality. Considering that $S(X_1,X_2,t)$ often has certain error $\kappa$ in computing results compared with $f_S(f_{X_1},f_{X_2},t)$, $S(X_1,X_2,t)$ depends on $X_1(t)+X_2(t)S(X_1,X_2,t-1)$ if for a certain error $\kappa$, $\exists t_1 > 0, \forall t >t_1, |S(X_1,X_2,t)-f_S(f_{X_1},f_{X_2},t)| < \kappa$ while for $\forall t < t_1, |S(X_1,X_2,t)-f_S(f_{X_1},f_{X_2},\\t)| > \kappa$. Because $S$ satisfy the sequential mapping of Definition 1, for the given $\kappa$ and $f_S(f_{X_1},f_{X_2}t)$, there exists a positive real number $\tau_S$ as a propagation delay such that $\forall t  (  (|S(X_1,X_2,t)-f_S(f_{X_1},f_{X_2},t)|>\kappa) \rightarrow (|S(X_1,X_2,t+\tau_S)-f_S(f_{X_1},f_{X_2},t+\tau_S)|<\kappa)  ) $. Let $t_2 = t_1+ \tau_S$, we have that $\forall t > t_2, |S(X_1,X_2,t)-f_S(f_{X_1},f_{X_2},t)| < \kappa$. Thus, $\forall t > t_2$, $S(t_2)$ depends on $X_1(t)+X_2(t)S(X_1,X_2,t-1)$. Therefore, the circuit $S$ preserves $t$-th Boolean signals from $t$-th inputs and $(t-1)$-th Boolean signals, it satisfies the definition of a sequential logic circuit\cite{katz2005contemporary}. Hence, circuits that satisfy sequential mapping and possess NOT-AND gate computing capabilities can implement sequential logic circuits.
\end{proof}

Following the formal delineation of sequential mapping and its intrinsic relationship with sequential logic circuits, this section advances to demonstrate the sequential mapping  of electronic logic circuits.

\subsection{Sequential Mapping of Electronic Logic Circuits}
In electronic computers, electronic logic gates represent “0” and “1” signals using high and low voltage levels, and their behavior depends on the physical components that constitute them. Circuits composed solely of cascaded electronic NOT-AND gates can equivalently implement all other Boolean logic functions. Thus, demonstrating that both the input-output mapping and cascading of the electronic NOT-AND gate satisfy sequential mapping suggests that electronic logic circuits may exhibit sequential mapping.

\begin{lemma}
    The input-output mapping of the electronic NOT-AND gate satisfies sequential mapping.
    \label{lem-1}
\end{lemma}

\begin{proof}
    The truth table for the electronic NOT-AND gate is presented in table \ref{tab1}. Let the inputs of the electronic NOT-AND gate $S$ at time $t$ be $X_1(t)$ and $X_2(t)$, and the output after the propagation delay $t_d$ be $Y(t+t_d)$, then we have $Y(t+t_d)=S(X_1,X_2,t+t_d)$. Based on Table \ref{tab1}, we can define a sequential mapping function $f_S$ for the electronic NOT-AND gate as $f_S(X_1,X_2,t)=1-X_1(t) X_2(t)$. According to the definition of the input-output mapping for the electronic NOT-AND gate, we have: $Y(t+t_d)=S(X_1,X_2,t+t_d)=1-X_1(t) X_2(t)$. Letting $\tau_S=t_d$, for any input interval greater than $\tau_S$, the input signals are constants, we have $X(t+\tau_S) = X(t)$. In this input interval, the output of the electronic NOT-AND gate satisfies: $S(X_1,X_2,t+\tau_S)=1-X_1(t)X_2(t)=1-X_1(t+\tau_S)X_2(t+\tau_S)=f_S(X_1,X_2,t+\tau_S)$. Letting $f_{X_1}=X_1$ and $f_{X_2}=X_2$, we have: $|S(X_1,X_2,t+\tau_S)-f_S(f_{X_1},f_{X_2},t+\tau_S)|=|S(X_1,X_2,t+\tau_S)-f_S(X_1,X_2,t+\tau_S)|=0<\kappa$. Thus, we have $\forall t  ( (|S(X_1,X_2,t)-f_S(f_{X_1},f_{X_2},t)|>\kappa) \rightarrow (|S(X_1,X_2,t+\tau_S)-f_S(f_{X_1},f_{X_2},t+\tau_S)|<\kappa) )$. It now suffices to show that the electronic NOT-AND gate satisfies sequential mapping. 
\end{proof}
\begin{table}[t!]
\centering
\captionsetup{width=\textwidth}
\caption{Truth table for the electronic NOT-AND gate}
\label{tab1}
\begin{tabular}{cccc}
\toprule
Input  & Input   & Output  & Sequential mapping function \\
$X_1(t)$ & $X_2(t)$ & $Y(t+t_d)$ & $f_S(X_1, X_2, t)$ \\
\midrule
        0 & 0 & 1 & 1\\
        0 & 1 & 1 & 1\\
        1 & 0 & 1 & 1\\
        1 & 1 & 0 & 0\\
\botrule
\end{tabular}
\end{table}

Since any Boolean logic function can be equivalently implemented using electronic NOT-AND gates and their cascades, we can prove sequential mapping of electronic logic circuits in two steps. First, we need to prove that the electronic NOT-AND logic gate satisfy the sequential mapping. Then, we can demonstrate that the electronic logic circuit also satisfies sequential mapping.
    
\begin{lemma}
If mappings $S$ and $G$ satisfy the sequential mapping, the cascaded mapping $S^\prime=G(S)$ also satisfies it. \label{lem-2}
\end{lemma}

\begin{proof}
 Let the input signals of circuit $S'$ be $X_1,\dots,X_n$, the output signals of circuits $S$ and $G$ at time $t$ are $S(X_1,\dots,X_n,t)$ and $G(S,X_1,\dots,X_n,t)$ respectively.
    Let the sequential mapping functions of circuits $S$ and $G$ be $f_S$ and $f_G$.
    From Definition 1, there exists $\tau_G > 0$ such that $\forall t((|G(S,X_1,\dots,X_n,t) - f_G(f_S,f_{X_1},\dots,f_{X_n},t) | > \kappa ) \to (|G(S,X_1,\dots,X_n,t+\tau_G) - f_G(f_S,f_{X_1},\dots,f_{X_n},t+\tau_G) | < \kappa ))$ holds. Since $S^\prime$ is the cascade of $S$ and $G$, we have: $S^\prime(X_1,\dots,X_n,t)=G(S,X_1,\dots,X_n,t)$ and $f_{S^\prime}(f_{X_1},\dots,f_{X_n},t) = f_G(f_S,f_{X_1},\dots,f_{X_n},t)$. Thus, there exists $\tau_{S^\prime} = \tau_G$ such that the output of the cascaded circuit $S^\prime$ satisfies: $\forall t((|S^\prime(X_1,\dots,X_n,t) - f_{S^\prime}(f_{X_1},\dots,f_{X_n},t) | > \kappa ) \to (|S^\prime(X_1,\dots,X_n,t+\tau_{S^\prime}) - f_{S^\prime}(f_{X_1},\dots,f_{X_n},t+\tau_{S^\prime}) | < \kappa )) $. According to Definition 1, the cascaded circuit $S^\prime=G(S)$ satisfies sequential mapping.
\end{proof}

From Lemma \ref{lem-1} and Lemma \ref{lem-2}, we have proved that both the electronic NOT-AND gate and its cascaded arrangements have sequential mapping. We will further demonstrate that other electronic logic circuits also satisfy this property based on these lemmas in the following sections.

\begin{theorem}
    Electronic logic circuits satisfies sequential mapping.
    \label{theorem-2}
\end{theorem}

\begin{proof}
    From Lemma \ref{lem-1}, we know that a single electronic NOT-AND gate satisfies sequential mapping. From Lemma \ref{lem-2}, we can conclude that the cascading of multiple electronic NOT-AND gates also meets sequential mapping. Moreover, any Boolean logic function can be equivalently realized using only electronic NOT-AND gates. Therefore, electronic logic circuits with Boolean logic function can be equivalently constructed using cascaded electronic NOT-AND gates and satisfy sequential mapping.
\end{proof}

\begin{corollary}
    Electronic circuits are capable of realizing sequential logic circuits.
    \label{lemma-3}
\end{corollary}

\begin{proof}
    From Theorem \ref{theorem-2}, the inputs and outputs of electronic circuits satisfy the sequential mapping. In addition, electronic circuits can realize the computational capability of OR gates and AND gates through semiconductor devices. It is proved in Theorem \ref{theorem-1} that circuits which satisfy sequential mapping while exhibiting the computational capability of NOT-AND gate effectively realize sequential logic circuits, thus electronic circuits have the ability to realize sequential logic circuits.
\end{proof}

Sequential logic is important in general-purpose computers. It provides general-purpose computers with the ability to store states, implement sequential control, and manage dynamic behaviors. For example, it enable fundamental operations in CPUs ——  instruction cycles, as well as data frame synchronization (e.g., UART, SPI) and error checking (e.g., CRC generators) in communication systems \cite{ashar1992sequential}. Without sequential logic circuits, electronic computers would degrade into simple calculators capable only of performing instantaneous operations, unable to handle any tasks requiring state memory or sequential coordination. They are the crucial bridge that transforms the digital world from "static" to "dynamic".

In this section, we proposed the concept of sequential mapping and demonstrated that sequential mapping is one of the sufficient conditions for sequential logic. In addition, it has been discussed that electronic computers satisfy this definition. Sequential logic circuits play an important role in general-purpose computers. The ability of electronic computers to implement sequential logic circuits makes them general-purpose computers. To construct a new type of general-purpose computer satisfying sequential mapping, we will present the biochemical computing mode in the following and prove its sequential mapping.

\section{Biochemical Computing Mode}
In this section, we will model, design, and analyze enzymatic reactions to create biochemical logic gates and demonstrate their sequential mapping. To be specific, we will propose logic gate models of enzymatic reactions and biochemical computing, including the construction and implementation of biochemical NOT, AND, and OR gates.

\subsection{Biochemical Gates Based on Enzymatic Reactions}
Enzymatic reactions refer to reactions that are accelerated by the catalytic action of enzymes \cite{koshland1958application}. Under the influence of enzymes, the conversion rate between substrates and products can be enhanced several and even millions of times, and enzymes are not consumed in the reaction process. Thus, we can utilize the non-consumable catalytic property of enzymes to realize biochemical logic gates with sequential mapping. The biochemical logic gates presented here are based on enzymatic reactions execute computational logic. In the enzymatic reaction process, reactants and products are referred to as substrates and products, respectively. Our proposed biochemical logic gates use small molecular compounds as substrates and products, with enzymes acting as the driving input signals. The enzymes as input signals are immobilized on solid supports to control the amount of enzyme and realize plug-and-play input signals. When the solid support containing the enzyme is immersed in a solution, i.e., the inserted enzyme, the input signal indicates “1”. Conversely, when the solid support is withdrawn from the solution, i.e., the extracted enzyme, the input signal is “0”. In enzyme-catalyzed reactions, the amount of substance remains constant and is not consumed, allowing it to be reused. The immobilization techniques facilitate the control of enzyme contact with the liquid reaction system. This approach addresses the issue of chemical signals being non-reusable during computation, which enables the construction of sustainable and reusable biochemical logic gates.

The Michaelis-Menten equation \cite{lopez2000generalized} describes the relationship between the initial velocity of an enzymatic reaction and substrate concentration, as seen in Equation \eqref{eq1}:
\begin{equation}
V_0 = V_{max}  \frac{[S]}{K_m+[S]} 
\label{eq1}
\end{equation}
where:
If A denotes a chemical substance, then $[A]$ denotes its relative concentration, i.e., the percentage of A with its value ranging from 0 to 1.
$[S]$ is the relative concentration of the substrate and $V_0$ is the initial velocity of the enzymatic reactions. 
$K_m$ is the Michaelis constant, representing the substrate concentration $[S]$ at which the enzyme catalysis reaches half of the maximum velocity $V_{max}$. When the substrate concentration is sufficiently high or approaches saturation, the reaction velocity gradually approaches the constant value $V_{max}$, which is defined by:
\begin{equation}
    V_{max}=k_{cat}\left[E\right]
    \label{eq2}
\end{equation}
where $k_{cat}$ is the catalytic constant of the enzyme, and $[E]$ is the relative concentration of the enzyme.

The Michaelis-Menten equation reveals the changes of substrate concentration as well as their influencing factors, which is crucial for modeling biochemical logic gates. When substrate is sufficient, the maximum reaction rate $V_{max}$ is influenced by both $[E]$ and $k_{cat}$. For a given enzyme, $k_{cat}$ is a constant, while enzyme concentration can be controlled based on the input-output characteristics of the logic gate. Thus, when the substrate concentration is sufficiently high or approaches saturation, we denote the maximum speed of enzyme $E$ as:
\begin{equation}
    {V_E=V}_{max}=k_{cat}\left[E\right]  
    \label{eq3}  
\end{equation}

The Boolean values of the input and output signals of the biochemical logic gates are defined by the threshold function $T(x)$:
\begin{eqnarray}
    T\left(x\right)=
    \left\{
    \begin{aligned}
        0,  & \quad x<\tau_0             \\
        1,  & \quad x> \tau_1
    \end{aligned}
    \right.
    \label{eqa1}
\end{eqnarray}
where $0<\tau_0<\tau_1<1$, and $x$ represents the concentration of the input or output substance. The constants $\tau_0$ and $\tau_1$ are the threshold values for logic 0 and 1. When the concentration is below $\tau_0$, the signal corresponds to logic 0; when the concentration is above $\tau_1$, it corresponds to logic 1. When the concentration is between $\tau_0$ and $\tau_1$, the signal is invalid, and the logic value cannot be identified. At this time, the biochemical computing component remains in the computation process until completion, at which point the output signal transitions to either logic 0 or logic 1.

The operation of biochemical logic gates consists of three processes: input, computation, and output. In the input process, the input of logic gates can be achieved by controlling input variables, i.e. enzyme concentration or activity. In the computation process, multiple enzymatic reactions interacting within the logic gate can perform the computational logic. In the output process, detection instruments can convert chemical concentration signals into observable signals.

\subsection{Biochemical NOT Gate}
The NOT gate is a fundamental unit of logic circuits. The biochemical computing mode will utilize two enzymes to catalyze the interconversion between a pair of substances to construct a biochemical NOT gate.

\subsubsection{Foundations of the Biochemical Reaction for the NOT Gate}

Suppose a biochemical reaction equation
\begin{equation}
S_1 \underset{P_1}{\overset{E_1}{\rightleftharpoons}} S_1^{\prime}
\label{eq5}
\end{equation}
where $S_1$ and $S_1^\prime$ denote small molecular compounds, and enzyme $E_1$ acts as the catalyst, facilitating the reaction to proceed more rapidly to the right. Additionally, there exists reaction as shown in Equation \eqref{eq5}, where $S_1^\prime$ can be converted back to $S_1$ under the catalysis of enzyme $P_1$.

According to the Michaelis-Menten equation \eqref{eq1}, the maximum reaction rate for positive reaction \eqref{eq5} under $E_1$ catalysis is $V_{E_1}$, while the maximum reaction rate for negative reaction \eqref{eq5} under $P_1$ catalysis is $V_{P_1}$. When the substrate concentration in the reaction is saturated and the catalytic activity of $P_1$ satisfies $V_{P_1}<V_{E_1}$, the catalytic capability of $P_1$ is less than that of $E_1$. This suggests that the rate of positive reaction is greater than that of negative reaction, while the concentrations of $S_1$ and $S_1^\prime$ are changing dynamically until both reactions in the system reach equilibrium. Thus, we can control the concentration changes of substrates $S_1$ and $S_1^\prime$  through the catalytic activity of $E_1$. When both $P_1$ and $E_1$ are present and can catalyze reactions, both positive reaction and negative reaction will occur. Since the intensity of positive reaction is greater than that of negative reaction, the consumption rate of $S_1$ is faster than its generation rate, while the generation rate of $S_1^\prime$ is faster than its consumption rate. Thus, relative to the initial concentrations of $S_1$ and $S_1^\prime$ , the concentration of $S_1^\prime$  is increased, while the concentration of $S_1$ is decreased. Conversely, when only $P_1$ is present, only negative reaction proceeds, resulting in a decrease in $S_1^\prime$  concentration and an increase in $S_1$ concentration. Thus, by controlling whether $E_1$ can catalyze the reaction and regulating its catalytic activity, we can regulate the concentrations of $S_1$ and $S_1^\prime$  within the reaction system. Based on this phenomenon, we can utilize the relationships between the changes in the concentrations of these enzymes and substrates to construct a biochemical NOT gate.

\subsubsection{Biochemical NOT Gate}
Based on the abovementioned model of the biochemical NOT gate, we can define its input, output, and mapping relationships.

\begin{definition}[Biochemical NOT gate]
    In reaction shown in Equation \eqref{eq5}, the mapping relationship between the concentration of $E_1 ([E_1 ])$ and the concentration of $S_1 ([S_1 ])$ constitutes a biochemical NOT gate. The input of the biochemical NOT gate is $[E_1 ]$, and the output is the $[S_1 ]$, satisfying $T([S_1])=\neg T([E_1])$. 
    \label{def-2}
\end{definition}

To achieve this input-output mapping relationship, the concentrations in the biochemical reaction model must meet the following Theorem \ref{theorem-3}.

\begin{theorem}
    When $E_1$ is not present, or when its concentration is higher than that of $P_1$ ($[P_1]>0$), the catalytic rates of two reactions satisfy $V_{P_1}<V_{E_1}$, and the input and output of the biochemical NOT gate satisfy $T([S_1])=\neg T([E_1])$.
\label{theorem-3}
\end{theorem}

\begin{proof}
    Let $V_3$ and $V_4$ represent the reaction rates of positive reaction and negative reaction in Equation \eqref{eq5} respectively. According to the Michaelis-Menten equation, the rates of positive reaction and negative reaction in Equation \eqref{eq5} can be expressed in relation to the substrate and enzyme concentrations as follows:
\begin{equation}
            V_3=\frac{V_{E_1}\left[S_1\right]}{K_{mE_1}+\left[S_1\right]}
            \label{eq6}
        \end{equation}
        \begin{equation}
            V_4=\frac{V_{P_1}\left[S_1^\prime\right]}{K_{mP_1}+\left[S_1^\prime\right]}
            \label{eq7}
\end{equation}
where $K_{mE_1}$ is the inherent catalytic constant of $E_1$, and $K_{mP_1}$ is the inherent constant of $P_1$. According to Equation \eqref{eq5}, let $V_{E_1}$ and $V_{P_1}$ represent the catalytic rates of $E_1$ and $P_1$ on these reactions. Based on the law of conservation of mass, the sum of the amounts of compounds $S_1$ and $S_1^\prime$ remains constant during the reaction. When the output signal of the logic gate stabilizes, the reaction reaches a chemical equilibrium state, satisfying $V_3=V_4$:  
\begin{equation}
            \frac{[E_1][S_1]}{K_{mE_1}+[S_1]}=\frac{[P_1](1-[S_1])}{K_{mP_1}+(1-[S_1])}
            \label{eq_NOT}
\end{equation}

To summarize, we can derive the relationship between the input and output of the biochemical NOT gate. As shown in Equation \eqref{eq_NOT}, we have: $\lim\limits_{[E_1]\to0}[S_1]=1$. When the relative concentration of the input substance $E_1$ corresponds to the input signal $T([E_1 ])$ being logical 0, the output signal $T([S_1 ])$ is logical 1. Conversely, according to Equation \eqref{eq_NOT}, we have $\lim\limits_{[E_1]\to 1}[S_1]=0$. When the relative concentration of input substance $E_1$ corresponds to the input signal $T([E_1 ])$ being logical 1, the output signal $T([S_1 ])$ is logical 0. Fig. \ref{fig_NOT} illustrates the relationship between the relative concentrations of input substance $E_1$ and output substance $S_1$. Table \uppercase\expandafter{\romannumeral2} shows the relationship between the input and output logical signal of the biochemical NOT gate. The relationship between input $T([E_1 ])$ and output $T([S_1 ])$, as shown in Table \uppercase\expandafter{\romannumeral2}, indicates that this logic gate satisfies the logical relationship of a NOT gate, i.e., $T([S_1 ])=\neg T([E_1 ])$.
\end{proof}
\begin{figure}[t]
        \centering
        \includegraphics[width=.4\linewidth]{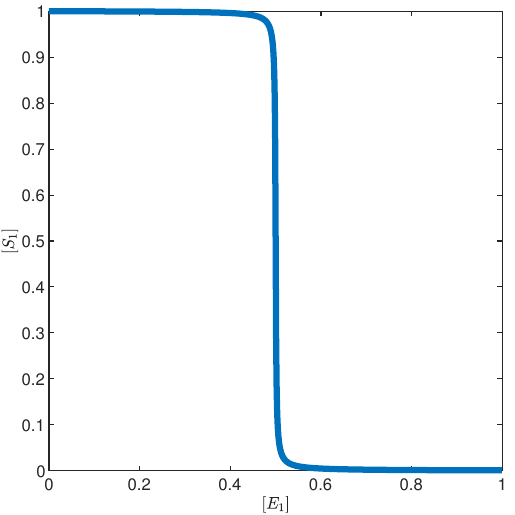}
        \caption{The input-output concentration relationship of biochemical NOT gate.}
        \label{fig_NOT}
    \end{figure}

\begin{table}[t]
\centering
\caption{Input-output mapping of the biochemical NOT gate}
\label{tab2}

\begin{tabular}{>{\centering\arraybackslash}p{3cm} >{\centering\arraybackslash}p{3cm} >{\centering\arraybackslash}p{3cm}}
\toprule
Input $T([E_1])$ & Output $T([S_1])$ & $\lnot T([E_1])$ \\
\midrule
        0 & 1 & 1 \\
        1 & 0 & 0 \\
\botrule
\end{tabular}
\end{table}

\subsection{Biochemical OR and AND Gates}
Unlike the biochemical NOT gate, which has a single input, the biochemical OR and AND gates have two inputs. Therefore, the biochemical computing model will utilize three enzymes to control substrate concentration changes, which can realize the simultaneous control of the logic gate’s output by both inputs. We can adjust substrate concentration by controlling enzyme concentration, thereby enabling the functionality of OR and AND gates.

\subsubsection{Foundations of the Biochemical Reaction for the OR Gate}
The biochemical OR gate operates through the catalytic action of three enzymes that facilitate the conversion of a pair of substances. Assume that there exists a small molecular compound $S_2$ satisfying the following reactions:
\begin{equation}
    S_2\stackrel{E_2}{\longrightarrow }S_2^\prime
    \label{eq8}
\end{equation}
\begin{equation}
    S_2\stackrel{E_3}{\longrightarrow }S_2^\prime
    \label{eq9}
\end{equation}
\begin{equation}
    S_2^\prime\stackrel{P_2}{\longrightarrow }S_2
    \label{eq10}
\end{equation}

The above reactions indicate that the conversion of $S_2$ to $S_2^\prime$ can be catalyzed by $E_2$ or $E_3$. According to the Michaelis-Menten equation, at sufficiently high substrate concentration, the maximum rate under the catalysis of $E_2$ for reaction \eqref{eq8} is $V_{E_2}$, the maximum rate under the catalysis of $E_3$ for reaction \eqref{eq9} is $V_{E_3}$, and the maximum rate under the catalysis of $P_2$ for reaction \eqref{eq10} is $V_{P_2}$.
\subsubsection{Biochemical OR Gate Computing Model}
Based on the biochemical reaction model of the biochemical OR gate, we can define its input, output, and mapping relationships as well.

\begin{definition}[Biochemical OR gate]
    In Equations \eqref{eq8}, \eqref{eq9}, and \eqref{eq10}, the concentrations of $E_2$, $E_3$, and $S_2^\prime$ establish the mapping of the biochemical OR gate. The inputs of the biochemical OR gate are the concentrations of $E_2 ([E_2 ])$ and $E_3 ([E_3 ])$. The output is the concentration of $S_2^\prime$ denoted as $[S_2^\prime]$ , satisfying the relationship: $T([S_2^\prime])=T([E_2])\lor T([E_3])$.
\end{definition}

To fulfill the input-output mapping relationship of the biochemical OR gate, the catalytic activities of enzymes in its biochemical reaction model must also meet the corresponding relationships.

\begin{theorem}
    When $E_2$ and $E_3$ are not present, or when their concentration is higher than that of $P_2$, the relationship of catalytic rates satisfies $V_{P_2}<V_{E_2}$ and $V_{P_2}<V_{E_3}$, and the input-output relationship of the biochemical logic OR gate satisfies: $T([S_2^\prime])=T([E_2])\lor T([E_3])$.
\label{theorem-4}
\end{theorem}

\begin{proof}
    Let the reaction rates of reactions in Equations \eqref{eq8}, \eqref{eq9}, and \eqref{eq10} be denoted as $V_7$, $V_8$, and $V_9$, respectively. According to the Michaelis-Menten equation, the relationship between the reaction rates of reactions \eqref{eq8}, \eqref{eq9}, and \eqref{eq10} and the concentrations of substrates and enzymes can be defined as:
\begin{equation}
            V_7=\frac{V_{E_2}\left[S_2\right]}{K_{mE_2}+\left[S_2\right]}
            \label{eq11}
\end{equation}
\begin{equation}
            V_8=\frac{V_{E_3}\left[S_2\right]}{K_{mE_3}+\left[S_2\right]}
            \label{eq12}
\end{equation}        
\begin{equation}
            V_9=\frac{V_{P_2}\left[S_2^\prime\right]}{K_{mP_2}+\left[S_2^\prime\right]}
            \label{eq13}
\end{equation}
When the reactions reach equilibrium, i.e., the concentrations of all substances being relatively stable, we have: $V_7+V_8=V_9$. This leads to the equation:
\begin{equation}
            \frac{[E_2](1-[S_2^\prime])}{K_{mE_2}+(1-[S_2^\prime])}
            +
            \frac{[E_3](1-[S_2^\prime])}{K_{mE_3}+(1-[S_2^\prime])}
            =
            \frac{[P_2][S_2^\prime]}{K_{mP_2}+[S_2^\prime]}
            \label{eq_AND_OR}
\end{equation}
When $V_{P_2}<V_{E_2}$ and $V_{P_2}<V_{E_3}$ are satisfied, we can know that the OR gate under all input combinations from equations \eqref{eq11}, \eqref{eq12}, \eqref{eq13} and \eqref{eq_AND_OR} satisfies the following relationships based on equation \eqref{eq_AND_OR}:
First, we have: $\lim\limits_{\left[E_2\right]\rightarrow0,\left[E_3\right]\rightarrow0\ }{\left[S_2^\prime\right]=0}$. When both inputs are 0, the output of the biochemical logic OR gate is 0. Second, we have $\lim\limits_{\left[E_2\right]\rightarrow0,\left[E_3\right]\rightarrow1\ }{\left[S_2^\prime\right]=1}$ and $\lim\limits_{\left[E_2\right]\rightarrow1,\left[E_3\right]\rightarrow0\ }{\left[S_2^\prime\right]=1}$. In two inputs $[E_2 ]$ and $[E_3 ]$, when one input is 0 and the other is 1, the output is 1. Third, we have $\lim\limits_{\left[E_2\right]\rightarrow1,\left[E_3\right]\rightarrow1\ }{\left[S_2^\prime\right]=1}$. When both inputs are 1, the output remains 1. Table \ref{tab3} illustrates the relationship between inputs and outputs of the biochemical logic OR gate. Thus, we can know that the relationship between the output $T([S_2^\prime])$ and the inputs $T([E_2])$ and $T([E_3])$ satisfies the OR gate: $T([S_2^\prime])=T([E_2])\land T([E_3])$.
\end{proof}

\begin{table}[t!]
\centering
\captionsetup{width=\textwidth}
\caption{The input-output mapping of biochemical logic OR gate}
\label{tab3}
\begin{tabular}{cccc}
\toprule
Input $T([E_2])$ & Input $T([E_3])$ & Output $T([S_1])$ & $T([E_2])\vee T([E3])$ \\
\midrule
        0 & 0 & 0 & 0\\
        0 & 1 & 1 & 1\\
        1 & 0 & 1 & 1\\
        1 & 1 & 1 & 1\\
\botrule
\end{tabular}
\end{table}

Hence, by utilizing the substances $S_2$, $S_2^\prime$, $E_2$, $E_3$ and $P_2$ from reactions \eqref{eq8}, \eqref{eq9} and \eqref{eq10}, and ensuring that the catalytic reaction rates $V_{E_2}$ and $V_{E_3}$ exceed $V_{P_2}$, we can construct a biochemical OR gate with $E_2$ and $E_3$ as input concentrations and $S_2^\prime$ as the output concentration. Furthermore, by modifying the relationships between $V_{E_2}$, $V_{E_3}$, and $V_{P_2}$, it is possible to achieve the input-output mapping for a biochemical logic AND gate.

\subsection{Biochemical AND Gate}
Both the biochemical AND and OR gates operate through the catalytic action of three enzymes that facilitate the interconversion of a pair of substances. At this time, the catalytic rate of each enzyme should satisfy: 
    \begin{equation}
        \begin{matrix}V_{P_2}>V_{E_2}
            \\V_{P_2}>V_{E_3}
            \\V_{P_2}<V_{E_2}+V_{E_3}
           
           \end{matrix}
           \label{eq14}
    \end{equation}
Observe that the total reaction rate of reactions \eqref{eq8} and \eqref{eq9} can only exceed that of reaction \eqref{eq10} when both $E_2$ and $E_3$ are fully active.

Since biochemical AND and OR gates share the same biochemical reaction model, we can define the input, output, and mapping relationships of the AND gate.

\begin{definition}[Biochemical AND gate]
    In reactions shown in Equations \eqref{eq8}, \eqref{eq9}, and \eqref{eq10}, when Equation \eqref{eq14} is satisfied, the concentrations of $E_2 ([E_2 ])$ and $E_3 ([E_3 ])$, and substrate  $S_2^\prime([S_2^\prime])$  follow the mapping relationship of a biochemical AND gate. The inputs for this gate are $[E_2 ]$ and $[E_3]$, while the output is $[S_2^\prime]$, satisfying: $T([S_2^\prime])=T([E_2])\land T([E_3])$.
    \label{def-4}
\end{definition}

Similar to the biochemical OR gate, in order to satisfy the input-output mapping of the AND gate, the concentrations of substances in the biochemical reaction model must fulfill the following conditions.

\begin{theorem}
    When the catalytic rates $V_{E_2}$, $V_{E_3}$, and $V_{P_2}$ satisfy Equation \eqref{eq14}, the input-output relationship of the biochemical AND gate satisfies: $T([S_2^\prime])=T([E_2])\land T([E_3])$.
    \label{theorem-5}
\end{theorem}

\begin{proof}
    Since the chemical equations for the AND gate are the same as those for the OR gate, the reaction rate equations can be expressed by Equations \eqref{eq11}, \eqref{eq12}, and \eqref{eq13}. However, given the inclusion of condition \eqref{eq14}, the input-output relationship for the biochemical AND gate can be derived as follows: First, given ${\lim\limits_{[E_2] \to 0,[E_3] \to 0\ }{[S_2^\prime]=0}}$, when both inputs are 0, the output of the biochemical logic AND gate is 0. Second, given $\lim\limits_{[E_2 ]\to 0,[E_3]\to 1\ }{[S_2^\prime]=0}$ and $\lim\limits_{[E_2]\rightarrow1,[E_3]\rightarrow0\ }{[S_2^\prime]=0}$, when only one input is 1 and the other is 0, the output is also 0. Third, given$\lim\limits_{\left[E_2\right]\rightarrow1,\left[E_3\right]\rightarrow1\ }{\left[S_2^\prime\right]=1}$, when both inputs are 1, the output is 1. These relationships can be summarized in Table \ref{tab4}. The relationship between the output $T([S_2^\prime])$ and the inputs $T([E_2 ])$ and $T([E_3 ])$ suggests that the biochemical AND gate satisfies the logical relationship of an AND gate: $T([S_2^\prime])=T([E_2])\land T([E_3])$.
\end{proof}

\begin{table}[t!]
\centering
\caption{The input-output mapping of biochemical AND gate}
\label{tab4}
\begin{tabular}{cccc}
\toprule
Input $T([E_2])$ & Input $T([E_3])$ & Output $T([S_2^\prime])$ & $T([E_2])\wedge T([E3])$ \\
\midrule
        0 & 0 & 0 & 0\\
        0 & 1 & 0 & 0\\
        1 & 0 & 0 & 0\\
        1 & 1 & 1 & 1\\
\botrule
\end{tabular}
\end{table}
Therefore, by utilizing the substances $S_2$, $S_2^\prime$, $E_2$, $E_3$ and $P_2$ from reactions shown in Equations \eqref{eq8}, \eqref{eq9}, and \eqref{eq10},~we can construct a biochemical AND gate with the concentrations of $E_2$ and $E_3$ as inputs and the concentration of $S_2^\prime$ as the output. In this configuration, $V_{E_2}$ and $V_{E_3}$ are less than $V_{P_2}$, while $V_{P_2}$ is less than the sum of $V_{E_2}$ and $V_{E_3}$.

\section{Sequential Mapping and Logic of Biochemical Computing Mode}
In the previous section, we demonstrated the reusability and static input-output characteristics of biochemical logic gates. We will further explore the dynamic input-output properties of biochemical logic gates by introducing the temporal dimension to prove the sequential mapping of biochemical computing mode. We will first prove that both the biochemical NOT and OR gate satisfy this property.

\begin{theorem}
    The biochemical NOT gate satisfies sequential mapping.
    \label{theorem-6}
\end{theorem}

\begin{proof}
    According to Theorem \ref{theorem-3}, the biochemical NOT gate comprises reactions \eqref{eq5}, with inputs and outputs represented by the concentrations of $E_1$ and $S_1$, respectively. Let the function of input $I$ of the biochemical NOT gate changing with the reaction time be $I(t)=[E_1 ]$, and the function of the output changing with the reaction time be $S(t)=[S_1 ]$. The sequential mapping function of the biochemical NOT gate is defined as the NOT relationship, which can be expressed as $f(t)=1-[E_1 ]$. It follows from Theorem \ref{theorem-3} that when the chemical reaction reaches equilibrium, the relationship between the two concentrations satisfies $[S_1]=\neg [E_1]$.
Let $0<\kappa<1$ be the output error of the biochemical NOT gate and let $t_0>0$ be any operating time of the gate. If the concentrations of substances meet the equilibrium condition at $t_0$, then $|S(t_0)-f_S(t_0)|\leqslant  \kappa$. Conversely, if the equilibrium condition is not met, then $|S(t_0)-f_S(t_0)|>\kappa$. This reaction will proceed toward equilibrium, causing the output $S$ to change over time. This relationship is defined by Equations \eqref{eq6} and \eqref{eq7}:
        \begin{equation}
            \frac{dS}{dt}=\frac{d[S_1]}{dt} = V_4-V_3 =
            \frac{[P_1](1-[S_1])}{K_m+(1-[S_1])}
            -
            \frac{[E_1][S_1]}{K_m+[S_1]}     
            \label{eq_dSdt}
        \end{equation}

The greater the difference between the substance concentration and the equilibrium state (i.e., $|S(t) - f(t)|$), the longer it takes to reach equilibrium. When $|S(t) - f(t)|= 1$, the reaction takes the longest time to reach equilibrium. Let the maximum time required to reach equilibrium be $t_{max}$ . Let the maximum reaction time be $t_+$ when the initial state is $S(t_0 )=0$ and $f(t_0 )=1$. Let the maximum reaction time be $t_-$ when the initial state is $S(t_0 )=1$ and $f(t_0 )=0$. Thus, we have: $t_{max}=\max(t_+,t_-)$. Where, $t_+$ and $t_-$ represent the longest times required for the reaction to reach equilibrium in both directions, while $t_{max}$  denotes the longest time required to reach equilibrium. When the initial state is $S(t_0 )=0$ and $f(t_0 )=1$, from Equation \eqref{eq_dSdt}, we have:
\begin{equation}
            \frac{dS}{dt}=\frac{d[S_1]}{dt} = V_4-V_3 =
            \frac{[P_1](1-S)}{K_m+(1-S)}
            >
            \frac{[P_1](1-S)}{K_m+1}
            >0
            \label{eq_dSdt_2}
        \end{equation}
        
Define the function $h(t): \mathbb{R} \to \mathbb{R} $ such that:
\begin{equation}
            \left\{\begin{matrix}
                \frac{dh}{dt}=\frac{[P_1](1-h)}{K_m+1}
                \\
                h(t_0)=0
            \end{matrix}
            \right.
            \label{eq_dhdt}
        \end{equation}
Since $h(t_0)=S(t_0)=0$ and 
$$\frac{dh}{dt}=\frac{[P_1](1-h)}{K_m+1}<\frac{[P_1](1-h)}{K_m+(1-h)}
        \leq \frac{[P_1](1-S)}{K_m+(1-S)}=\frac{dS}{dt},$$
        it follows that $ h(t)<S(t)$ for any $t>t_0$.
Solving Equation \eqref{eq_dhdt}, we get:
\begin{equation}
            h(t) = 1 - \exp(-\frac{[P_1](t-t_0)}{K_m+1})
        \end{equation}
It is observed that the maximum reaction time $t_+$ satisfies $h(t_0+t_+)=1-\kappa$, leading to: $t_+= - \frac{Km+1}{[P_1]}\ln(\kappa)>0$. Since the multiple input of biochemical gate exists time intervals, we have $f_S(t_0+t_+)=f_S(t_0)=1$. Hence: $|S(t_0+t_+)-f_S(t_0+t_+)|=1-S(t_0+t_+)<1-h(t_0+t_+)<\kappa$.

Moreover, when the initial state is $S(t_0 )=1$ and $f_S(t_0 )=0$, from Equation \eqref{eq_dSdt}, we have:
\begin{equation}
\begin{split}
    \frac{dS}{dt}&=\frac{d[S_1]}{dt} = V_4-V_3 \\
    &= \frac{[P_1](1-[S_1])}{K_m+(1-[S_1])}
    -
    \frac{[E_1][S_1]}{K_m+[S_1]}
    <
    [P_1]-\frac{[S_1]}{K_m+1}
    <0
\end{split}
\end{equation}

Define the function $l(t):\mathbb{R} \to \mathbb{R} $, such that:
\begin{equation}
            \left\{\begin{matrix}
                \frac{dl}{dt}=[P_1]-\frac{h}{K_m+1}
                \\
                h(t_0)=1
            \end{matrix}
            \right.
            \label{eq_dldt}
\end{equation}
Given that $l(t_0 )=S(t_0 )=1$, we have:
\begin{align*}
\frac{dl}{dt} &= [P_1] - \frac{l}{K_m + 1} \nonumber \\
&> \frac{[P_1](1 - l)}{K_m + (1 - l)} - \frac{[E_1]l}{K_m + l} \nonumber \\
&\geqslant \frac{[P_1](1 - [S_1])}{K_m + (1 - [S_1])} - \frac{[E_1][S_1]}{K_m + [S_1]} \nonumber \\
&= \frac{dS}{dt}
\end{align*}
Thus, $l(t)>S(t)$ holds for any $t>t_0$. Solving Equation \eqref{eq_dldt} yields:
\begin{equation}
    l(t) = (1+K_m)[P_1]+\exp(\frac{t_0-t}{1+K_m})
\end{equation}
In this case, the maximum reaction time $t_-$ satisfies $h(t_0+t_-)=\kappa$, leading to:
$$t_-= - \frac{\ln(\kappa-(1+K_m)[P_1])}{1+K_m}>0$$
Given that there are time intervals in multiple inputs of the logic gate, we have $f_S(t_0+t_- )=f_S(t_0 )=0$. Then, 
$$|S(t_0+t_-)-f_S(t_0+t_-)|=S(t_0+t_-)<l(t_0+t_-)<\kappa$$

In sum, the maximum time required for the biochemical NOT gate to reach equilibrium is $t_{max}=\max(t_+,t_-)$, which satisfies: if at $t_0$ the concentration of substances meets $|S(t_0)-f_S(t_0)|>\kappa$, then $$|S(t_0+t_{max})-f_S(t_0+t_{max})|<\kappa.$$ Thus, there exists $$\tau_S = t_{max} = \max(
            - \frac{Km+1}{[P_1]}\ln(\kappa)
            ,
            - \frac{\ln(\kappa-(1+K_m)[P_1])}{1+K_m}
            )$$
which lets the output $S$ of the biochemical NOT gate satisfy: $$\forall t  (  (|S(t)-f_S(t)|>\kappa) \rightarrow (|S(t+\tau_S)-f_S(t+\tau_S)|<\kappa)  ) $$
According to Definition 1, we can conclude that the biochemical NOT gate satisfies sequential mapping.
\end{proof}

Since all Boolean logic can be equivalently realized using only NOT and AND gates, we need to prove sequential mapping of the biochemical AND gate to further demonstrate sequential mapping of the biochemical computer.

\begin{theorem}
    The biochemical AND gate satisfies sequential mapping.
    \label{theorem-7}
\end{theorem}
\begin{proof}
    According to Theorem \ref{theorem-5}, the biochemical AND gate includes reactions in Equations \eqref{eq8}, \eqref{eq9}, and \eqref{eq10}, with inputs being the concentrations of $E_2([E_2 ])$ and $E_3([E_3 ])$ and the output being the concentration of $S_2^\prime([S_2^\prime])$. Let the functions of the two inputs of the biochemical AND gate, which vary with reaction time, be defined as $I_1(t)=[E_2]$ and $I_2(t)=[E_3]$, and the output as $S(I_1,I_2,t)=[S_2^\prime]$. The sequential mapping function of the biochemical AND gate is given by the logical AND relationship: $f_S(f_{I_1},f_{I_2}, t)=[E_2] [E_3]$. For convenience of discussion, in the following $S(I_1,I_2,t)$ is abbreviated as $S(t)$ and $f_S(f_{I_1},f_{I_2},t)$ as $f_S(t)$. Thus, $S(t)=[S_2^\prime]$ and $f_S(t) = [E_2] [E_3]$ hold. As seen in Theorem 5, when the chemical reaction reaches equilibrium, the concentrations satisfy: $\left[S_2^\prime\right]=[E_2] [E_3]$. Let the output error $\kappa$ be a positive number. If time $t$ satisfies chemical reaction equilibrium, then: $|S(t)-f_S(t)|\leqslant \kappa$. If $t$ does not satisfy equilibrium, the reaction will proceed towards equilibrium, causing changes in the output $S$. The relationship of the output over time $S(t)$ is defined by the Equations \eqref{eq11}, \eqref{eq12}, and \eqref{eq13}. From these equations, we can conclude that the chemical reaction can reach equilibrium within a given time $\tau_0$, which is defined by $\kappa$ and the reaction parameters $K_m$ and $V_{max}$ in Equations \eqref{eq11}, \eqref{eq12}, and \eqref{eq13}. If time $t$ does not satisfy equilibrium, we have: $|S(t+\tau_0)-f_S(t+\tau_0)|<\kappa$. Thus, there exists $\tau_S=\tau_0>0$ such that the output $S$ of the biochemical AND gate satisfies: $\forall t  (  (|S(t)-f_S(t)|>\kappa) \rightarrow (|S(t+\tau_S)-f_S(t+\tau_S)|<\kappa)  ) $. According to Definition 1 the biochemical AND gate satisfies sequential mapping.
\end{proof}

Given that both the biochemical NOT and AND gates have sequential mapping, we conclude that their cascade also satisfies this property based on Lemma \ref{lem-2}. Since all Boolean logic based on biochemical computation can be implemented through the cascade of biochemical NOT and AND gates, the biochemical computing mode also demonstrates sequential mapping.

\begin{theorem}
    Biochemical circuits satisfy sequential mapping.
    \label{theorem-8}
\end{theorem}
\begin{proof}
    Based on Theorems \ref{theorem-6} and \ref{theorem-7}, the biochemical NOT and AND gates satisfy sequential mapping. Two biochemical circuits that satisfy the sequential mapping are cascaded to obtain a new circuit, forming a new cascade mapping. According to Lemma \ref{lem-2}, the new cascade mapping also satisfies the sequential mapping. Since all Boolean functions can be constructed by cascading mappings using only NOT gates and AND gates, biochemical circuits satisfy sequential mapping as well.
\end{proof}
\begin{theorem}
    Biochemical logic circuits are capable of realizing sequential logic circuits.
\end{theorem}
\begin{proof}
    Theorem \ref{theorem-3} and Theorem \ref{theorem-5} respectively demonstrate that biochemical logic circuits possess the computational capabilities of NOT and AND gates. Theorem \ref{theorem-8} establishes that biochemical logic circuits satisfy sequential mapping. Therefore, based on Theorem \ref{theorem-2}, it can be concluded that biochemical logic circuits are capable of realizing sequential logic circuits.
\end{proof}

In conclusion, biochemical logic gates based on the biochemical computing mode not only meet the truth tables of logic gates in static input-output characteristics, but also satisfy sequential mapping in dynamic input-output characteristics. Therefore, biochemical logic gates can produce correct dynamic outputs when receiving dynamic inputs, which demonstrates the capability to continuously process input sequences. This feature solves the problem of irreversible substrate consumption in chemical signaling, which provides a foundation for realizing biochemical sequential logic circuits. The foundation of a general-purpose Turing machine is the state machine—a sequential logic circuit that operates continuously, responding to dynamically changing signals \cite{sep-turing-machine}. Hence, the biochemical computing mode proposed in this study is of great importance for the development of a general-purpose biochemical reaction computer.

\section{Conclusion}
This paper proposed the concept of sequential mapping, which is one of the important prerequisites for realizing sequential logic circuits. Sequential logic circuits are a critical condition for the computing model to become a general-purpose computer, which can ensure the continuous repeated computing and storage capabilities of the computer. Focusing on sequential logic, this paper designed a new computing mode and carriers(e.g., substrates and enzymes) from a biochemical perspective. Considering the characteristics of enzymes such as non-consumption, specificity, and high-efficiency catalysis, biochemical computing mode proposed in this paper based on enzymatic reactions also exhibits the aforementioned potential advantages. Meanwhile, this paper demonstrated that biochemical computing mode satisfies sequential mapping, similar to electronic computers. Consequently, from a theoretical perspective, the proposed mode introduces a novel carrier for computer science and can further implement other biochemical computational elements, constructing biochemical reaction-based computers with general computing capabilities.

While this paper has not yet conducted specific biochemical reaction experiments, we theoretically demonstrate that biochemical computing mode satisfies sequential mapping, verifying its feasibility. With the continuous advancement of science and technology, future research will further explore and realize realistic biochemical logic circuits and prototype devices based on biochemical computing mode.

\section*{Acknowledgments}
We gratefully acknowledge the financial support of Guangdong Taiyi High-tech Development Co., Ltd.

\bibliography{BMBbib}

\begin{thebibliography}{10}
\expandafter\ifx\csname url\endcsname\relax
  \def\url#1{\burl{#1}}\fi
\expandafter\ifx\csname urlprefix\endcsname\relax\def\urlprefix{URL }\fi
\providecommand{\bibinfo}[2]{#2}
\providecommand{\eprint}[2][]{\url{#2}}
\providecommand{\doi}[1]{\url{https://doi.org/#1}}
\bibcommenthead

\bibitem{von1993first}
\bibinfo{author}{Von~Neumann, J.}
\newblock \bibinfo{title}{First draft of a report on the edvac}.
\newblock \emph{\bibinfo{journal}{IEEE Annals of the History of Computing}} \textbf{\bibinfo{volume}{15}}, \bibinfo{pages}{27--75} (\bibinfo{year}{1993}).

\bibitem{markov2014limits}
\bibinfo{author}{Markov, I.~L.}
\newblock \bibinfo{title}{Limits on fundamental limits to computation}.
\newblock \emph{\bibinfo{journal}{Nature}} \textbf{\bibinfo{volume}{512}}, \bibinfo{pages}{147--154} (\bibinfo{year}{2014}).

\bibitem{katz2005contemporary}
\bibinfo{author}{Katz, R.~H.} \& \bibinfo{author}{Borriello, G.}
\newblock \emph{\bibinfo{title}{Contemporary logic design}}  (\bibinfo{publisher}{Pearson Prentice Hall}, \bibinfo{address}{Upper Saddle River, NJ}, \bibinfo{year}{2005}).

\bibitem{burks1946preliminary}
\bibinfo{author}{Burks, A.~W.}, \bibinfo{author}{Goldstine, H.~H.} \& \bibinfo{author}{Von~Neumann, J.}
\newblock \emph{\bibinfo{title}{Preliminary discussion of the logical design of an electronic computing instrument}}, \bibinfo{pages}{399--413} (\bibinfo{publisher}{Springer}, \bibinfo{year}{1982}).

\bibitem{mcmahon2023physics}
\bibinfo{author}{McMahon, P.~L.}
\newblock \bibinfo{title}{The physics of optical computing}.
\newblock \emph{\bibinfo{journal}{Nature Reviews Physics}} \textbf{\bibinfo{volume}{5}}, \bibinfo{pages}{717--734} (\bibinfo{year}{2023}).

\bibitem{fan2020propelling}
\bibinfo{author}{Fan, D.}, \bibinfo{author}{Wang, J.}, \bibinfo{author}{Wang, E.} \& \bibinfo{author}{Dong, S.}
\newblock \bibinfo{title}{Propelling dna computing with materials’ power: Recent advancements in innovative dna logic computing systems and smart bio-applications}.
\newblock \emph{\bibinfo{journal}{Advanced Science}} \textbf{\bibinfo{volume}{7}}, \bibinfo{pages}{2001766} (\bibinfo{year}{2020}).

\bibitem{zhu2024quantum}
\bibinfo{author}{Zhu, H.} \emph{et~al.}
\newblock \bibinfo{title}{Quantum computing and machine learning on an integrated photonics platform}.
\newblock \emph{\bibinfo{journal}{Information}} \textbf{\bibinfo{volume}{15}}, \bibinfo{pages}{95} (\bibinfo{year}{2024}).

\bibitem{ambs2010optical}
\bibinfo{author}{Ambs, P.}
\newblock \bibinfo{title}{Optical computing: A 60-year adventure}.
\newblock \emph{\bibinfo{journal}{Advances in Optical Technologies}} \textbf{\bibinfo{volume}{2010}}, \bibinfo{pages}{372652} (\bibinfo{year}{2010}).

\bibitem{xu2020scalable}
\bibinfo{author}{XiaoYun, X.} \emph{et~al.}
\newblock \bibinfo{title}{A scalable photonic computer solving the subset sum problem}.
\newblock \emph{\bibinfo{journal}{Science advances}} \textbf{\bibinfo{volume}{6}}, \bibinfo{pages}{eaay5853} (\bibinfo{year}{2020}).

\bibitem{vazquez2018optical}
\bibinfo{author}{V{\'a}zquez, M.~R.} \emph{et~al.}
\newblock \bibinfo{title}{Optical np problem solver on laser-written waveguide platform}.
\newblock \emph{\bibinfo{journal}{Optics Express}} \textbf{\bibinfo{volume}{26}}, \bibinfo{pages}{702--710} (\bibinfo{year}{2018}).

\bibitem{wright2022deep}
\bibinfo{author}{Wright, L.~G.} \emph{et~al.}
\newblock \bibinfo{title}{Deep physical neural networks trained with backpropagation}.
\newblock \emph{\bibinfo{journal}{Nature}} \textbf{\bibinfo{volume}{601}}, \bibinfo{pages}{549--555} (\bibinfo{year}{2022}).

\bibitem{adleman1994molecular}
\bibinfo{author}{Adleman, L.~M.}
\newblock \bibinfo{title}{Molecular computation of solutions to combinatorial problems}.
\newblock \emph{\bibinfo{journal}{science}} \textbf{\bibinfo{volume}{266}}, \bibinfo{pages}{1021--1024} (\bibinfo{year}{1994}).

\bibitem{okamoto2004dna}
\bibinfo{author}{Okamoto, A.}, \bibinfo{author}{Tanaka, K.} \& \bibinfo{author}{Saito, I.}
\newblock \bibinfo{title}{Dna logic gates}.
\newblock \emph{\bibinfo{journal}{Journal of the American Chemical Society}} \textbf{\bibinfo{volume}{126}}, \bibinfo{pages}{9458--9463} (\bibinfo{year}{2004}).

\bibitem{Li2001dna}
\bibinfo{author}{Ren-Hou, L.} \& \bibinfo{author}{Wen, Y.}
\newblock \bibinfo{title}{Basic principles and discussions on dna computing}.
\newblock \emph{\bibinfo{journal}{Chinese Journal of Computers}} \textbf{\bibinfo{volume}{24}}, \bibinfo{pages}{972--978} (\bibinfo{year}{2001}).

\bibitem{song2019fast}
\bibinfo{author}{Song, T.} \emph{et~al.}
\newblock \bibinfo{title}{Fast and compact dna logic circuits based on single-stranded gates using strand-displacing polymerase}.
\newblock \emph{\bibinfo{journal}{Nature nanotechnology}} \textbf{\bibinfo{volume}{14}}, \bibinfo{pages}{1075--1081} (\bibinfo{year}{2019}).

\bibitem{chatterjee2017spatially}
\bibinfo{author}{Chatterjee, G.}, \bibinfo{author}{Dalchau, N.}, \bibinfo{author}{Muscat, R.~A.}, \bibinfo{author}{Phillips, A.} \& \bibinfo{author}{Seelig, G.}
\newblock \bibinfo{title}{A spatially localized architecture for fast and modular dna computing}.
\newblock \emph{\bibinfo{journal}{Nature nanotechnology}} \textbf{\bibinfo{volume}{12}}, \bibinfo{pages}{920--927} (\bibinfo{year}{2017}).

\bibitem{song2017renewable}
\bibinfo{author}{Song, X.}, \bibinfo{author}{Eshra, A.}, \bibinfo{author}{Dwyer, C.} \& \bibinfo{author}{Reif, J.}
\newblock \bibinfo{title}{Renewable dna seesaw logic circuits enabled by photoregulation of toehold-mediated strand displacement}.
\newblock \emph{\bibinfo{journal}{RSC advances}} \textbf{\bibinfo{volume}{7}}, \bibinfo{pages}{28130--28144} (\bibinfo{year}{2017}).

\bibitem{lv2023dna}
\bibinfo{author}{Lv, H.} \emph{et~al.}
\newblock \bibinfo{title}{Dna-based programmable gate arrays for general-purpose dna computing}.
\newblock \emph{\bibinfo{journal}{Nature}} \textbf{\bibinfo{volume}{622}}, \bibinfo{pages}{292--300} (\bibinfo{year}{2023}).

\bibitem{Shan2021dna}
\bibinfo{author}{Yang, S.}, \bibinfo{author}{Li, J.}, \bibinfo{author}{Cui, Y.} \& \bibinfo{author}{Teng, Y.}
\newblock \bibinfo{title}{The current status and future prospects of dna computing}.
\newblock \emph{\bibinfo{journal}{Chinese Journal of Biotechnology}} \textbf{\bibinfo{volume}{37}}, \bibinfo{pages}{1120--1130} (\bibinfo{year}{2021}).

\bibitem{vishweshwaraiah2021two}
\bibinfo{author}{Vishweshwaraiah, Y.~L.}, \bibinfo{author}{Chen, J.}, \bibinfo{author}{Chirasani, V.~R.}, \bibinfo{author}{Tabdanov, E.~D.} \& \bibinfo{author}{Dokholyan, N.~V.}
\newblock \bibinfo{title}{Two-input protein logic gate for computation in living cells}.
\newblock \emph{\bibinfo{journal}{Nature communications}} \textbf{\bibinfo{volume}{12}}, \bibinfo{pages}{6615} (\bibinfo{year}{2021}).

\bibitem{ravichandran2021efficient}
\bibinfo{author}{Ravichandran, D.} \emph{et~al.}
\newblock \bibinfo{title}{An efficient medical image encryption using hybrid dna computing and chaos in transform domain}.
\newblock \emph{\bibinfo{journal}{Medical \& biological engineering \& computing}} \textbf{\bibinfo{volume}{59}}, \bibinfo{pages}{589--605} (\bibinfo{year}{2021}).

\bibitem{feynman2018simulating}
\bibinfo{author}{Feynman, R.~P.}
\newblock \emph{\bibinfo{title}{Simulating physics with computers}}, \bibinfo{pages}{133--153} (\bibinfo{publisher}{cRc Press}, \bibinfo{address}{Boca Raton, FL}, \bibinfo{year}{2018}).

\bibitem{deutsch1985quantum}
\bibinfo{author}{Deutsch, D.}
\newblock \bibinfo{title}{Quantum theory, the church--turing principle and the universal quantum computer}.
\newblock \emph{\bibinfo{journal}{Proceedings of the Royal Society of London. A. Mathematical and Physical Sciences}} \textbf{\bibinfo{volume}{400}}, \bibinfo{pages}{97--117} (\bibinfo{year}{1985}).

\bibitem{chamberland2022building}
\bibinfo{author}{Chamberland, C.} \emph{et~al.}
\newblock \bibinfo{title}{Building a fault-tolerant quantum computer using concatenated cat codes}.
\newblock \emph{\bibinfo{journal}{PRX Quantum}} \textbf{\bibinfo{volume}{3}} (\bibinfo{year}{2022}).

\bibitem{alexeev2021quantum}
\bibinfo{author}{Alexeev, Y.} \emph{et~al.}
\newblock \bibinfo{title}{Quantum computer systems for scientific discovery}.
\newblock \emph{\bibinfo{journal}{PRX quantum}} \textbf{\bibinfo{volume}{2}}, \bibinfo{pages}{017001} (\bibinfo{year}{2021}).

\bibitem{kim2023evidence}
\bibinfo{author}{Kim, Y.} \emph{et~al.}
\newblock \bibinfo{title}{Evidence for the utility of quantum computing before fault tolerance}.
\newblock \emph{\bibinfo{journal}{Nature}} \textbf{\bibinfo{volume}{618}}, \bibinfo{pages}{500--505} (\bibinfo{year}{2023}).

\bibitem{pezzagna2021quantum}
\bibinfo{author}{Pezzagna, S.} \& \bibinfo{author}{Meijer, J.}
\newblock \bibinfo{title}{Quantum computer based on color centers in diamond}.
\newblock \emph{\bibinfo{journal}{Applied Physics Reviews}} \textbf{\bibinfo{volume}{8}} (\bibinfo{year}{2021}).

\bibitem{wlazlowski2024fermionic}
\bibinfo{author}{Wlaz{\l}owski, G.}, \bibinfo{author}{Forbes, M.~M.}, \bibinfo{author}{Sarkar, S.~R.}, \bibinfo{author}{Marek, A.} \& \bibinfo{author}{Szpindler, M.}
\newblock \bibinfo{title}{Fermionic quantum turbulence: Pushing the limits of high-performance computing}.
\newblock \emph{\bibinfo{journal}{PNAS nexus}} \textbf{\bibinfo{volume}{3}}, \bibinfo{pages}{pgae160} (\bibinfo{year}{2024}).

\bibitem{daley2022practical}
\bibinfo{author}{Daley, A.~J.} \emph{et~al.}
\newblock \bibinfo{title}{Practical quantum advantage in quantum simulation}.
\newblock \emph{\bibinfo{journal}{Nature}} \textbf{\bibinfo{volume}{607}}, \bibinfo{pages}{667--676} (\bibinfo{year}{2022}).

\bibitem{chen2023neuromorphic}
\bibinfo{author}{Chen, X.} \emph{et~al.}
\newblock \bibinfo{title}{Neuromorphic photonic memory devices using ultrafast, non-volatile phase-change materials}.
\newblock \emph{\bibinfo{journal}{Advanced Materials}} \textbf{\bibinfo{volume}{35}}, \bibinfo{pages}{2203909} (\bibinfo{year}{2023}).

\bibitem{bradley2019ten}
\bibinfo{author}{Bradley, C.~E.} \emph{et~al.}
\newblock \bibinfo{title}{A ten-qubit solid-state spin register with quantum memory up to one minute}.
\newblock \emph{\bibinfo{journal}{Physical Review X}} \textbf{\bibinfo{volume}{9}}, \bibinfo{pages}{031045} (\bibinfo{year}{2019}).

\bibitem{burks1981first}
\bibinfo{author}{Burks, A.~W.} \& \bibinfo{author}{Burks, A.~R.}
\newblock \bibinfo{title}{First general-purpose electronic computer}.
\newblock \emph{\bibinfo{journal}{Annals of the History of Computing}} \textbf{\bibinfo{volume}{3}}, \bibinfo{pages}{310--389} (\bibinfo{year}{1981}).

\bibitem{adleman1998computing}
\bibinfo{author}{Adleman, L.~M.}
\newblock \bibinfo{title}{Computing with dna}.
\newblock \emph{\bibinfo{journal}{Scientific american}} \textbf{\bibinfo{volume}{279}}, \bibinfo{pages}{54--61} (\bibinfo{year}{1998}).

\bibitem{chaplin1990enzyme}
\bibinfo{author}{Chaplin, M.~F.} \& \bibinfo{author}{Bucke, C.}
\newblock \emph{\bibinfo{title}{Enzyme technology}}  (\bibinfo{publisher}{CUP Archive}, \bibinfo{address}{Cambridge, UK}, \bibinfo{year}{1990}).

\bibitem{sep-turing-machine}
\bibinfo{author}{De~Mol, L.}
\newblock \bibinfo{title}{ in \textit{{Turing Machines}}} \bibinfo{edition}{{W}inter 2024} edn, (eds \bibinfo{editor}{Zalta, E.~N.} \& \bibinfo{editor}{Nodelman, U.}) \emph{\bibinfo{booktitle}{The {Stanford} Encyclopedia of Philosophy}}  (\bibinfo{publisher}{Metaphysics Research Lab, Stanford University}, \bibinfo{year}{2024}).

\bibitem{ashar1992sequential}
\bibinfo{author}{Ashar, P.}, \bibinfo{author}{Devadas, S.} \& \bibinfo{author}{Newton, A.~R.}
\newblock \emph{\bibinfo{title}{Sequential logic synthesis}}  (\bibinfo{publisher}{Springer Science \& Business Media}, \bibinfo{address}{New York}, \bibinfo{year}{1992}).

\bibitem{koshland1958application}
\bibinfo{author}{Koshland~Jr, D.~E.}
\newblock \bibinfo{title}{Application of a theory of enzyme specificity to protein synthesis}.
\newblock \emph{\bibinfo{journal}{Proceedings of the National Academy of Sciences}} \textbf{\bibinfo{volume}{44}}, \bibinfo{pages}{98--104} (\bibinfo{year}{1958}).

\bibitem{lopez2000generalized}
\bibinfo{author}{Lopez, S.} \emph{et~al.}
\newblock \bibinfo{title}{A generalized michaelis-menten equation for the analysis of growth}.
\newblock \emph{\bibinfo{journal}{Journal of animal science}} \textbf{\bibinfo{volume}{78}}, \bibinfo{pages}{1816--1828} (\bibinfo{year}{2000}).

\end{thebibliography}

\end{document}